\documentclass[11pt]{article}
\usepackage[margin=1in]{geometry}
\usepackage{amsmath, amssymb, amsthm, amsfonts}
\usepackage{graphicx}
\usepackage{hyperref}
\usepackage[table]{xcolor}
\usepackage{colortbl}
\usepackage{diagbox}
\usepackage{svg}
\usepackage{float}
\usepackage{enumitem}
\usepackage[skip=0pt,labelformat=simple,labelsep=none]{subcaption}

\newtheorem{theorem}{Theorem}
\newtheorem{lemma}[theorem]{Lemma}

\graphicspath{{./figures/},{./packings/}}

{\makeatletter
	\gdef\xxxmark{%
		\expandafter\ifx\csname @mpargs\endcsname\relax 
		\expandafter\ifx\csname @captype\endcsname\relax 
		\marginpar{xxx}
		\else
		xxx 
		\fi
		\else
		xxx 
		\fi}
	\gdef\xxx{\@ifnextchar[\xxx@lab\xxx@nolab}
	\long\gdef\xxx@lab[#1]#2{\textbf{[\xxxmark #2 ---{\sc #1}]}}
	\long\gdef\xxx@nolab#1{\textbf{[\xxxmark #1]}}
}

\let\realbfseries=\bfseries
\def\bfseries{\realbfseries\boldmath}

\title{When Can You Tile an Integer Rectangle with Integer Squares?}

\author{%
	MIT CompGeom Group%
	\thanks{%
		Artificial first author to highlight that the other authors
		(in alphabetical order) worked as an equal group.
		Please include all authors (including this one) in your bibliography,
		and refer to the authors as ``MIT CompGeom Group'' (without ``et al.'').}
\and
	Zachary Abel%
	\thanks{Massachusetts Institute of Technology,
	  \protect\url{{zabel,edemaine,jaysonl}@mit.edu}}
\and
	Hugo A. Akitaya%
	\thanks{University of Massachusetts, Lowell, \protect\url{hugo_akitaya@uml.edu}}
\and
	Erik D. Demaine\footnotemark[2]
\and
	Adam C. Hesterberg%
	\thanks{Harvard University, \protect\url{ahesterberg@seas.harvard.edu}}
\and
	Jayson Lynch\footnotemark[2]
}

\date{}

\makeatletter
\newcommand{\sq}[2]{%
	\if\relax\detokenize{#2}\relax
	$#1 \times #1$%
	\else
	$#1 \times #2$%
	\fi
}
\makeatother

\begin{document}
	\maketitle
	
	\begin{abstract}
		This paper characterizes when an $m \times n$ rectangle, where $m$ and $n$ are integers, can be tiled (exactly packed) by squares where each has an integer side length of at least $2$. 
		In particular, we prove that tiling is always possible when both $m$ and $n$ are sufficiently large (at least 10).
		When one dimension $m$ is small, the behavior is eventually periodic in $n$ with period 1, 2, or 3.
		When both dimensions $m,n$ are small, the behavior is determined computationally by an exhaustive search.
	\end{abstract}
	
	\section{Introduction}
	
	
	Tilings of rectangles by squares (``squared rectangles'') date back to Dehn in 1903 \cite{dehn1903zerlegung}, who proved that the sides are all \emph{commensurable} (their lengths have rational ratios).
	Thus we can scale such tilings to have vertices on an integer grid, with horizontal and vertical edges of integer length.
	In 1940, Brooks, Smith, Stone, and Tutte \cite{brooks1940dissection} famously studied \emph{perfect squarings} --- tilings of rectangles by \emph{distinct} squares. They refined Dehn's result and proved a converse --- every rectangle with commensurable sides has an infinite family of perfect tilings --- by showing a connection to electrical flow in circuits \cite{brooks1940dissection}.

	Without the perfect (distinctness) constraint, the converse of Dehn's result is obvious: any rectangle with commensurable sides, scaled to have integer dimensions $m \times n$, has a trivial tiling by $mn$ \sq11 squares.
	In this paper, we study one simple approach to preventing this trivial tiling: forbid \sq11 squares.

	Another well-studied approach is to specify a set of square sizes, or more generally rectangle dimensions, that are allowed in the tiling of a rectangle.
	In 1D --- tiling an interval with specified intervals --- all sufficiently large intervals can be tiled, and finding this cut-off size is the famous Frobenius coin/stamp problem \cite{coin-problem,bower2006packing,labrousse2010tiling}.
	The 2D problem was first studied in 1969 by de Bruijn \cite{de1969filling}, who characterized which rectangles can be tiled by a single rectangle.
	In 1978, Erd\H{o}s \cite{erdos-problem} posed two more interesting 2D problems --- allowing two square sizes or three rectangle sizes --- with solutions claimed (but never published) by E. G. Straus.
	Extensive work \cite{fricke1995,bower2004can,bower2006packing,labrousse2010tiling} has developed necessary and sufficient conditions for tiling a given rectangle with given rectangles (each of which can be used arbitrarily many times), as well as higher-dimensional generalizations (packing boxes or tori with boxes).
	In particular, this work formalizes the original claims by Straus, including that any three squares with coprime side lengths can tile all sufficiently large rectangles \cite{bower2006packing,labrousse2010tiling} --- a fact we use here (and reproduce, in a special case).



	If we specify the \emph{number} of squares of each size that must tile a given rectangle (or even square), then the problem is NP-complete \cite{Leung-Tam-Wong-Young-Chin-1990}.
	If we allow arbitrary squares but instead require the \emph{fewest} possible squares,
	Kenyon \cite{kenyon1996tiling} proved upper and lower bounds that match up to constant factors, which Walters \cite{walters2009rectangles} generalized to higher dimensions, while other work considers computational heuristics \cite{monaci2018minimum}.
	Heubach \cite{heubach1999tiling} studied \emph{counting} the number of different rectangle tilings when limiting the square sizes.

	\subsection{Our Results}
	\label{sec:characterization}

	In this paper, we characterize which rectangles of integer side lengths can be tiled by squares each of integer side length at least $2$, as follows:
	\begin{enumerate}[label=(\Roman*)]
		\item \label{char:even} \sq2n and \sq4n rectangles can tiled exactly when $n$ is even.
		\item \label{char:3} \sq3n rectangles can be tiled exactly when $n \equiv 0 \pmod 3$.
		\item \label{char:large} \sq mn rectangles for all $m \geq 5, n \geq 20$ (and symmetrically, $m \geq 20, n \geq 5$) can be tiled.
		\item \label{char:small} Table~\ref{tab:code_output} specifies tileability for all remaining $m,n$ (indeed, for all $m,n < 20$).
	\end{enumerate}
	In particular, Table~\ref{tab:code_output} indicates successful tilings for all $10 \leq m,n \leq 20$, so combined with \ref{char:large}, we obtain that the \sq mn rectangle is tileable without \sq11 squares for all $m,n \geq 10$.

	\begin{table}
		\centering
		\def\T{\checkmark}
		\tabcolsep=0.25em
		\newlength\margin
		\settowidth{\margin}{0}
		\margin=0.5\margin
		\def\M{\hspace*{\margin}}
		\begin{tabular}{>{\columncolor{purple!25}}rcccccccccccccccccc}
		\rowcolor{purple!25}
		\diagbox[width=1.45em,height=1.15em]{\!$_m$}{\clap{$^n$}}
		  &\M2\M&\M3\M&\M4\M&\M5\M&\M6\M&\M7\M&\M8\M&\M9\M&10&11&12&13&14&15&16&17&18&19\\
		 2&\T&  &\T&  &\T&  &\T&  &\T&  &\T&  &\T&  &\T&  &\T&  \\
		 3&  &\T&  &  &\T&  &  &\T&  &  &\T&  &  &\T&  &  &\T&  \\
		 4&\T&  &\T&  &\T&  &\T&  &\T&  &\T&  &\T&  &\T&  &\T&  \\
		 5&  &  &  &\T&\T&  &  &  &\T&\T&\T&  &  &\T&\T&\T&\T&  \\
		 6&\T&\T&\T&\T&\T&\T&\T&\T&\T&\T&\T&\T&\T&\T&\T&\T&\T&\T\\
		 7&  &  &  &  &\T&\T&  &  &\T&  &\T&\T&\T&  &\T&\T&\T&\T\\
		 8&\T&  &\T&  &\T&  &\T&  &\T&  &\T&  &\T&\T&\T&\T&\T&\T\\
		 9&  &\T&  &  &\T&  &  &\T&\T&  &\T&\T&\T&\T&\T&\T&\T&\T\\
		10&\T&  &\T&\T&\T&\T&\T&\T&\T&\T&\T&\T&\T&\T&\T&\T&\T&\T\\
		11&  &  &  &\T&\T&  &  &  &\T&\T&\T&\T&\T&\T&\T&\T&\T&\T\\
		12&\T&\T&\T&\T&\T&\T&\T&\T&\T&\T&\T&\T&\T&\T&\T&\T&\T&\T\\
		13&  &  &  &  &\T&\T&  &\T&\T&\T&\T&\T&\T&\T&\T&\T&\T&\T\\
		14&\T&  &\T&  &\T&\T&\T&\T&\T&\T&\T&\T&\T&\T&\T&\T&\T&\T\\
		15&  &\T&  &\T&\T&  &\T&\T&\T&\T&\T&\T&\T&\T&\T&\T&\T&\T\\
		16&\T&  &\T&\T&\T&\T&\T&\T&\T&\T&\T&\T&\T&\T&\T&\T&\T&\T\\
		17&  &  &  &\T&\T&\T&\T&\T&\T&\T&\T&\T&\T&\T&\T&\T&\T&\T\\
		18&\T&\T&\T&\T&\T&\T&\T&\T&\T&\T&\T&\T&\T&\T&\T&\T&\T&\T\\
		19&  &  &  &  &\T&\T&\T&\T&\T&\T&\T&\T&\T&\T&\T&\T&\T&\T\\
		\end{tabular}
		\caption{Which integer \sq mn rectangles, for $2 \leq m,n \leq 19$, admit tilings with squares of side length at least~$2$. \checkmark\ indicates when a tiling was found by brute force \cite{github}. Figure~\ref{fig:code_tilings} shows the found tilings.}
		\label{tab:code_output}
	\end{table}
	
	Our tilings use only \sq22, \sq33, \sq55, and \sq77 squares, so our result can also be cast in terms of restricting the set of allowed square sizes to these four.

	\section{Characterization}
	The proof of our result is a combination of mathematical arguments for rectangles with at least one side length large enough, and brute-force computer search for the remaining finitely many cases.
	
	First, observe that any even-by-even rectangle can be tiled by \sq22 squares. Thus we need only focus on cases where at least one side length is odd. 
	
	\subsection{Thin Rectangles: \sq2n, \sq4n \ref{char:even}, and \sq3n \ref{char:3}}
	Next we consider what turn out to be the only infinite families of rectangles which cannot be tiled.
	
	A \sq2n rectangle can only be tiled with \sq22 squares (as smaller are forbidden, and larger would not fit), and thus can only be tiled for $n$ even.
	
	For \sq3n rectangles, if we attempt to use a \sq22 square, then there will be a narrow \sq12 region which is impossible to fill. Thus \sq3n squares can only be exactly covered by \sq33 squares and thus can only be tiled when $n \equiv 0 \pmod 3$.
	
	Similarly, we cannot use \sq33 squares in a tiling of a $4\times n$ rectangle, and thus these rectangles are only tileable when $n$ is even.

	\subsection{Large Rectangles \ref{char:large}}
	By contrast, once we reach $m \geq 5$, \sq mn rectangles can be tiled for sufficiently large~$n$.
	Our arguments here are similar to known constructions for any three squares with coprime side lengths \cite{bower2006packing,labrousse2010tiling}, but taking care to exactly compute the dimensions needed for tileability by the three smallest such squares.
	
	\begin{lemma} \label{lem:5xn}
		A \sq5n rectangle can be tiled by \sq22, \sq33, and \sq55 squares when $n \geq 20$,
		or when $n \in \{5,6,10,11,12,15,16,17,18\}$.
	\end{lemma}
	\begin{proof}
		We can construct a \sq56 rectangle from two \sq33 squares and three \sq22 squares, as shown in Figure~\ref{fig:code_tilings} (top left).
		We also have the trivial tiling of a \sq55 rectangle with one square.
		By combining a sequence of these two tilings, we can tile any \sq5n rectangle where $n = 5 i + 6 j$ is an integer linear combination with $i,j \geq 0$.
		By Sylvester's solution \cite{sylvester1882} to the Frobenius 2-coin problem \cite{coin-problem,bower2006packing,labrousse2010tiling}, given that $\gcd(5,6) = 1$, this is possible for all $n > 5 \cdot 6 - 5 - 6 = 19$.
		It can be checked by hand that the remaining values $n \in \{5,6,10,11,12,15,16,17,18\}$ can be expressed as integer linear combinations of $5$ and $6$; see also the tilings in Figure~\ref{fig:code_tilings} (top row).
	\end{proof}
	
	Next we use the ability to shave off $5\times n$ rectangles to reduce a large enough rectangle down to an even-by-even rectangle, which can then be tiled by \sq22 squares.
	
	\begin{lemma} \label{lem:20x20}
	An $m \times n$ rectangle with $m, n \geq 20$ can be tiled by \sq22, \sq33, and \sq55 squares.
	\end{lemma}
	
	\begin{proof}
		First, if side lengths $m$ and $n$ are both even, then we can tile the rectangle with \sq22 squares.

		Next, if $m$ is odd, then we apply Lemma~\ref{lem:5xn} to tile a $5 \times n$ rectangle on the top, leaving an \sq{(m-5)}n rectangle to tile  where $m-5$ is even. If $n$ is even, we are done, via \sq22 squares.

		Finally, if $n$ is also odd, then we again apply Lemma~\ref{lem:5xn} to tile a \sq{(m-5)}{5} rectangle on the left. Note that $m-5 \geq 15$ and even, so Lemma~\ref{lem:5xn} applies. This leaves us with an \sq{(m-5)}{(n-5)} rectangle to tile, which is even by even, so we can pack it with \sq22 squares.
	\end{proof}
	
	The remaining large cases are \sq mn rectangles where $m \in [6, 19]$.

	\begin{theorem}
	An $m \times n$ rectangle with $m \geq 6$ and $n \geq 20$ can be tiled by \sq22, \sq33, \sq55, and \sq77 squares.
	\end{theorem}
	
	\begin{proof}
		If $m \geq 20$, then we are done by Lemma~\ref{lem:20x20}.
		If $m$ is even, then we apply Lemma~\ref{lem:5xn} to tile an $m \times 5$ rectangle on the top, leaving an \sq m{(n-5)} rectangle to tile, where $m$ and $n-5$ are even so we can use \sq22 squares.
		Thus we can assume $m$ is odd and $7 \leq m \leq 19$.

		Next we provide a way to extend tilings to increase $n$ by multiples of~$6$.
		We can tile an \sq m6 rectangle (for any odd $m \geq 5$) by taking the \sq56 tiling and adding rows of \sq22 squares.
		Thus, if we have a tiling of an \sq mn rectangle, then we can extend it to an \sq m{(n+6i)} rectangle by adding $i \geq 0$ copies of this \sq m6 tiling on the right.

		Now consider the finite cases in Table~\ref{tab:code_output}.
		If we have a row of six \checkmark s in a row starting at \sq mn, then we have tilings for \sq m{(n+i)} rectangles for $i \in \{0,1,2,3,4,5\}$.
		By the above extension, we obtain tilings for \sq m{n'} rectangles for all $n' \geq n$.
		Indeed, rows 9, 11, 13, 15, 17, and 19 of Table~\ref{tab:code_output} have such a sequence of six \checkmark s.
		Row 7 has only four \checkmark s at the end, but the row continues with two more: \sq 7{20} can be tiled by doubling the \sq 7{10} packing, and \sq 7{21} can be tiled by \sq77 squares.
		Thus, in all cases, we obtain \sq mn tilings for all $n \geq 20$.
		The tilings use at most \sq77 squares because the original tilings did; see Figure~\ref{fig:code_tilings} and Section~\ref{sec:small}.
	\end{proof}

	

	\subsection{Small Rectangles \ref{char:small}}
	\label{sec:small}

	To resolve the remaining cases, we implemented a brute-force algorithm to test which rectangles can be packed with integer squares of side length at least 2.
	The code is available as open source \cite{github}.
	We ran this algorithm on all rectangles from \sq22 up to \sq{19}{19}.
	Table~\ref{tab:code_output} shows the binary results: was a tiling found to be possible (\checkmark) or impossible (empty)?
	Figure~\ref{fig:code_tilings} shows the found tilings, up to symmetry; to reduce the number of tilings from 129 to 59, we exclude rectangles whose two side lengths have a common divisor $d \geq 2$, which have an easy packing by $d \times d$ squares.
	All tilings use only \sq22, \sq33, \sq55, and \sq77 squares. (Although the search was not restricted to these squares, it preferred them by trying squares in order of increasing size.)

	This concludes our proof of the characterization of Section~\ref{sec:characterization}.

	Amusingly, tilings for all but one of these rectangles can be found by the following divide-and-conquer approach.
	Start from the one-square tilings of the \sq22, \sq33, \sq55, and \sq77 squares (by themselves).
	Repeatedly join two found smaller tilings along a matching edge length.
	The only exception is the \sq{11}{13} rectangle, which has a tiling that cannot be built in this way.
	Of course, we still need brute force to verify that no other rectangles can be tiled.
	
	\begin{figure}
		\centering
		\def\PARSE#1x#2y{\def\nx{#1}\def\ny{#2}}
		\def\OUT#1{%
			\PARSE#1y%
			\subcaptionbox{\sq{\nx}{\ny}}{\includegraphics[scale=0.45]{code/output/#1}}%
		}

		\OUT{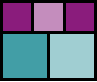}
		\OUT{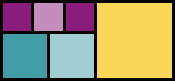}
		\OUT{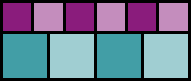}
		\OUT{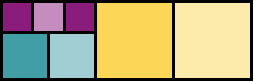}
		\OUT{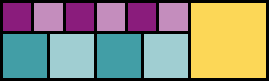}
		\OUT{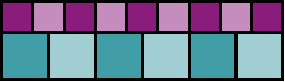}

		\OUT{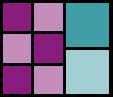}
		\OUT{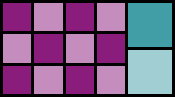}
		\OUT{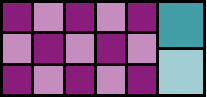}
		\OUT{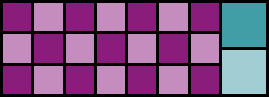}
		\OUT{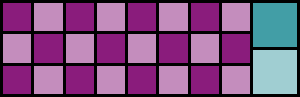}

		\OUT{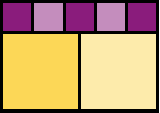}
		\OUT{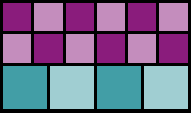}
		\OUT{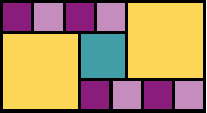}
		\OUT{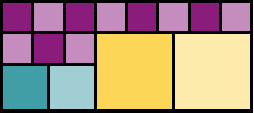}
		\OUT{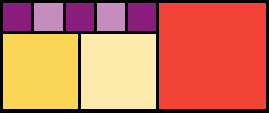}
		\OUT{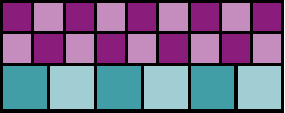}
		\OUT{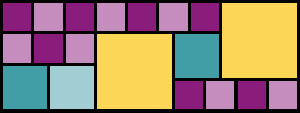}

		\OUT{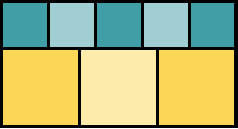}
		\OUT{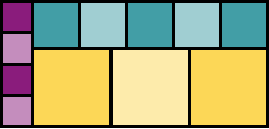}
		\OUT{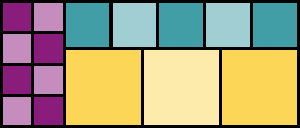}

		\OUT{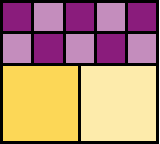}
		\OUT{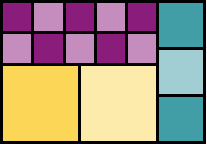}
		\OUT{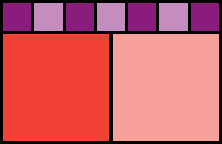}
		\OUT{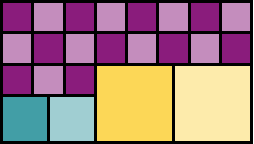}
		\OUT{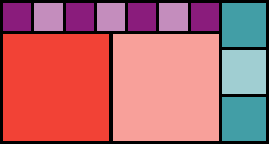}
		\OUT{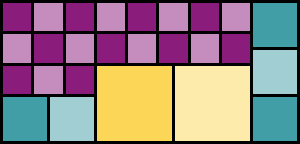}

		\OUT{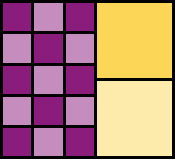}
		\OUT{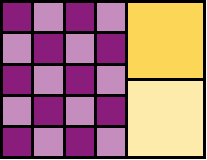}
		\OUT{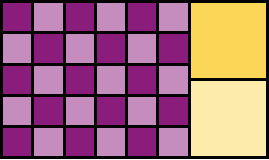}
		\OUT{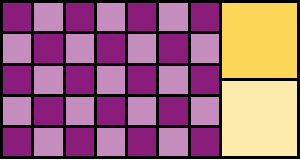}

		\OUT{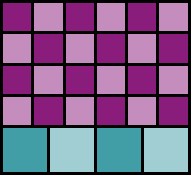}
		\OUT{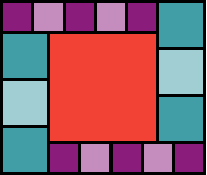}
		\OUT{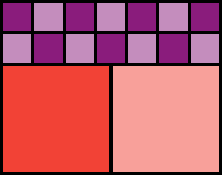}
		\OUT{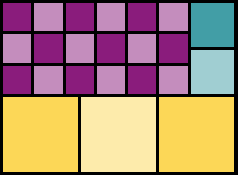}
		\OUT{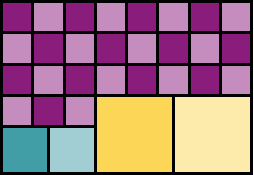}
		\OUT{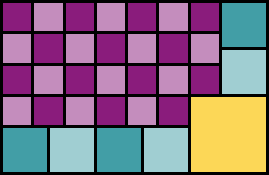}
		\OUT{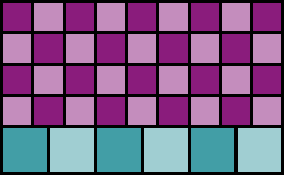}
		\OUT{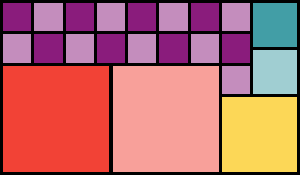}

		\OUT{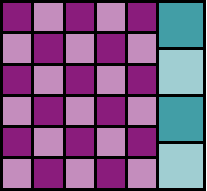}
		\OUT{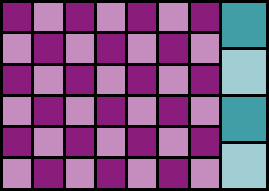}
		\OUT{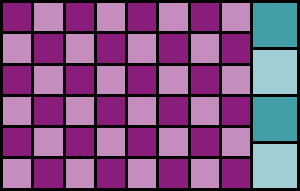}

		\OUT{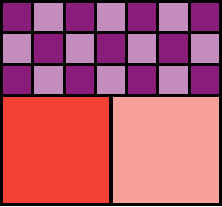}
		\OUT{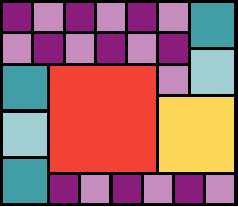}
		\OUT{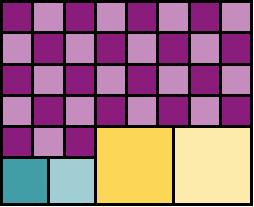}
		\OUT{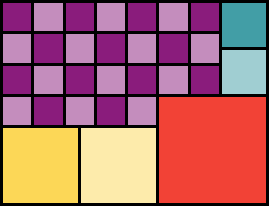}
		\OUT{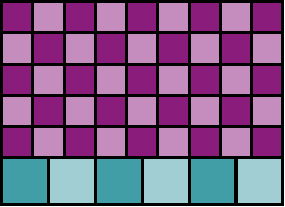}
		\OUT{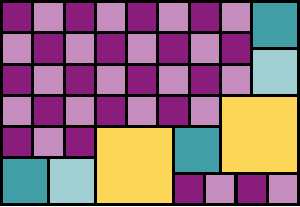}

		\OUT{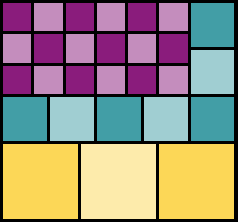}
		\OUT{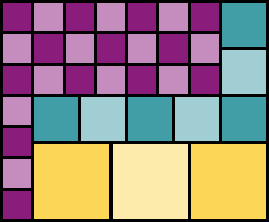}
		\OUT{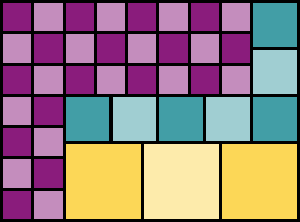}
		\qquad
		\OUT{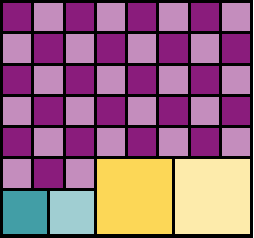}
		\OUT{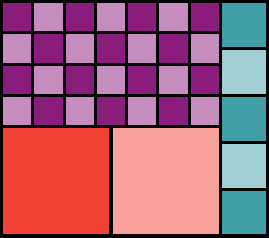}
		\OUT{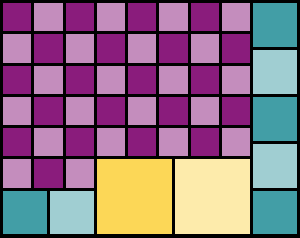}

		\OUT{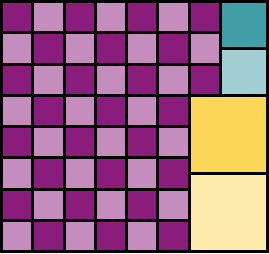}
		\OUT{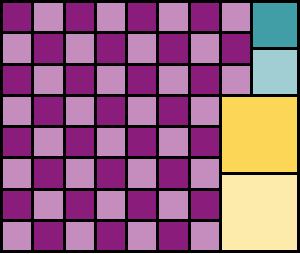}
		\qquad
		\OUT{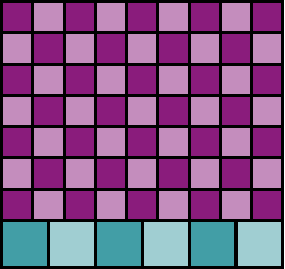}
		\OUT{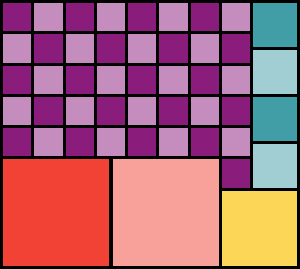}
		\qquad
		\OUT{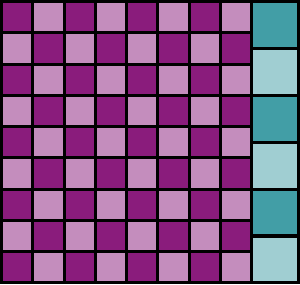}
		\caption{Tilings found by brute force \cite{github}, corresponding to \checkmark s in
			Table~\ref{tab:code_output} except for dimensions with a common factor.
			\sq22, \sq33, \sq55, and \sq77 squares are purple, teal, yellow, and red, respectively.}
		\label{fig:code_tilings}
	\end{figure}

	\section{Open Problems}
	Our tilings use only \sq22, \sq33, \sq55, and \sq77 squares, and our tilings for $m,n \geq 20$ do not need \sq77 squares.
	But our tilings for \sq 7n, \sq {13}n, and \sq{17}n rectangles seem to need \sq77 infinitely often; can we prove this?
	
	One could also consider counting the number of solutions. Heubach studied this problem for for up to \sq5n rectangles when \sq11 squares are allowed in the tiling~\cite{heubach1999tiling}.

	Another generalization of our problem is to higher dimensions.
	For example, in 3D, what $x \times y \times z$ integer boxes can be tiled by $k \times k \times k$ integer cubes with $k \geq 2$?
	Such tilings are known to exist if $x$, $y$, and $z$ are all sufficiently large \cite{bower2006packing,labrousse2010tiling}, but what about when one or more dimensions are small?
	
	\let\realbibitem=\bibitem
	\def\bibitem{\par \vspace{-1.2ex}\realbibitem}

	\bibliographystyle{alphakey}
	\bibliography{main}{}

\newcommand{\etalchar}[1]{$^{#1}$}
\begin{thebibliography}{LTW{\etalchar{+}}90}

\bibitem[BM04]{bower2004can}
Richard~J. Bower and T.~S. Michael.
\newblock When can you tile a box with translates of two given rectangular
  bricks?
\newblock {\em The Electronic Journal of Combinatorics}, 11(1):N7, 2004.

\bibitem[BM06]{bower2006packing}
Richard~J. Bower and T.~S. Michael.
\newblock Packing boxes with bricks.
\newblock {\em Mathematics Magazine}, 79(1):14--30, 2006.

\bibitem[BSST40]{brooks1940dissection}
R.~L. Brooks, C.~A.~B. Smith, A.~H. Stone, and W.~T. Tutte.
\newblock The dissection of rectangles into squares.
\newblock {\em Duke Mathematical Journal}, 7, December 1940.

\bibitem[dB69]{de1969filling}
Nicholas~G. de~Bruijn.
\newblock Filling boxes with bricks.
\newblock {\em The American Mathematical Monthly}, 76(1):37--40, 1969.

\bibitem[Deh03]{dehn1903zerlegung}
Max Dehn.
\newblock {\"U}ber {Z}erlegung von {R}echtecken in {R}echtecke.
\newblock {\em Mathematische Annalen}, 57(3):314--332, September 1903.

\bibitem[Fri95]{fricke1995}
Jan Fricke.
\newblock Quadratzerlegung eines {R}echtecks {(Partitioning a rectangle into
  squares, in German)}.
\newblock {\em Mathematische Semesterberichte}, 42(1):53--62, 1995.

\bibitem[Heu99]{heubach1999tiling}
Silvia Heubach.
\newblock Tiling an m-by-n area with squares of size up to k-by-k (m$\le$ 5).
\newblock {\em Congressus Numerantium}, 140:43--64, 1999.

\bibitem[HS78]{erdos-problem}
D.~A. Holton and Jennifer Seberry.
\newblock Problems.
\newblock In {\em Combinatorial Mathematics}, pages 346--349, Berlin,
  Heidelberg, 1978. Springer Berlin Heidelberg.
\newblock Problem 4: Tiling (P. Erdős).

\bibitem[Ken96]{kenyon1996tiling}
Richard Kenyon.
\newblock Tiling a rectangle with the fewest squares.
\newblock {\em Journal of Combinatorial Theory, Series A}, 76(2):272--291,
  1996.

\bibitem[LA10]{labrousse2010tiling}
D.~Labrousse and J.~L.~Ram{\'\i}rez Alfons{\'\i}n.
\newblock A tiling problem and the {F}robenius number.
\newblock In David Chudnovsky and Gregory Chudnovsky, editors, {\em Additive
  Number Theory: Festschrift in Honor of the Sixtieth Birthday of Melvyn B.
  Nathanson}, pages 203--220. Springer, 2010.

\bibitem[LTW{\etalchar{+}}90]{Leung-Tam-Wong-Young-Chin-1990}
Joseph Y.-T. Leung, Tommy~W. Tam, C.~S. Wong, Gilbert~H. Young, and Francis
  Y.~L. Chin.
\newblock Packing squares into a square.
\newblock {\em Journal of Parallel and Distributed Computing}, 10(3):271--275,
  1990.

\bibitem[MdS18]{monaci2018minimum}
Michele Monaci and Andr{\'e}~Gustavo dos Santos.
\newblock Minimum tiling of a rectangle by squares.
\newblock {\em Annals of Operations Research}, 271:831--851, 2018.

\bibitem[MIT23]{github}
MIT{ CompGeom Group}, Zachary Abel, Hugo~A. Akitaya, Erik~D. Demaine, Adam~C.
  Hesterberg, and Jayson Lynch.
\newblock Tiling integer rectangles with integer squares.
\newblock GitHub repository, 2023.
\newblock \url{https://github.com/MIT-CompGeom/tiling-rectangles-with-squares}.

\bibitem[Syl82]{sylvester1882}
J.~J. Sylvester.
\newblock On subvariants, i.e.\ semi-invariants to binary quantics of an
  unlimited order.
\newblock {\em American Journal of Mathematics}, 5(1):79--136, 1882.

\bibitem[Wal09]{walters2009rectangles}
Mark Walters.
\newblock Rectangles as sums of squares.
\newblock {\em Discrete Mathematics}, 309(9):2913--2921, 2009.

\bibitem[Wik23]{coin-problem}
Wikipedia.
\newblock Coin problem.
\newblock \url{https://en.wikipedia.org/wiki/Coin_problem}, 2023.

\end{thebibliography}

\end{document}